\documentclass[11pt,a4paper]{article}
\usepackage{fullpage}

\usepackage{amsmath,amssymb}
\usepackage{authblk}
\usepackage{color}
\usepackage[linesnumbered,vlined,ruled]{algorithm2e}

\newtheorem{theorem}{Theorem}

\newtheorem{lemma}{Lemma}

\newcommand{\qed}{\hfill $\Box$ \medbreak}
\newenvironment{proof}{\noindent {\bf Proof.}}{\qed}

\newcommand{\id}{\mbox{\rm id}}
\newcommand{\E}{\mathbf{E}}
\renewcommand{\P}{{\cal P}}
\newcommand{\Congest}{{\rm \textsf{CONGEST}}}
\newcommand{\Local}{{\rm \textsf{LOCAL}}}
\newcommand{\seq}{{\rm \textsf{Seq}}}
\newcommand{\dist}{{\rm \textsf{Dist}}}
\newcommand{\ppr}{p^{1-\frac{\log \log \log p + 4}{\log \log p}}}

\begin{document}

\title{Distributed Testing of Excluded Subgraphs}

\author[1]{Pierre Fraigniaud\thanks{Additional support from ANR project DISPLEXITY, and Inria project GANG. }}
\author[2]{Ivan Rapaport\thanks{Additional support from CONICYT via Basal in Applied Mathematics, N\'ucleo Milenio Informaci\'on y Coordinaci\'on en Redes ICM/FIC RC130003, Fondecyt 1130061.}}
\author[2]{Ville Salo\thanks{Additional support from FONDECYT project 3150552.}}
\author[3]{Ioan Todinca}

\affil[1]{CNRS and University Paris Diderot, France}
\affil[2]{DIM-CMM (UMI 2807 CNRS), Universidad de Chile, Chile}
\affil[3]{Universit\'e d'Orl\'eans, France}

\date{}

\maketitle

\begin{abstract}
We study \emph{property testing} in the context of \emph{distributed computing}, under the classical \Congest\ model. It is known that testing whether a graph is triangle-free can be done in a constant number of rounds, where the constant depends on how far  the input graph is from being triangle-free. We show that, for every connected 4-node  graph $H$, testing whether a graph is $H$-free can  be done in a constant number of rounds too. The constant also depends on how far the input graph is  from being $H$-free, and the dependence is identical to the one in the case of testing triangles. Hence, in particular, testing whether a graph is $K_4$-free, and  testing whether a graph is $C_4$-free can be done in a constant number of rounds (where $K_k$ denotes the $k$-node clique, and $C_k$ denotes the $k$-node cycle). On the other hand, we show that testing $K_k$-freeness and $C_k$-freeness for $k\geq 5$ appear to be much harder. Specifically, we investigate two natural types of generic  algorithms for testing $H$-freeness, called DFS tester and BFS tester. The latter captures the previously known algorithm to test the presence of triangles, while the former captures our generic algorithm to test the presence of a 4-node graph pattern $H$. We prove that both DFS and BFS testers fail to test $K_k$-freeness and $C_k$-freeness in a constant number of rounds for $k\geq 5$. 

\medskip

\noindent\textbf{Keywords:} property testing, decision and verification algorithms, \Congest\ model.
\end{abstract}

\vfill

\thispagestyle{empty}
\pagebreak
\setcounter{page}{1}
 
\section{Introduction}

Let $\P$ be a graph property, and let $0<\epsilon <1$ be a fixed parameter. According to the usual definition from  \emph{property testing}~\cite{Goldreich2010}, an $n$-node $m$-edge graph $G$ is $\epsilon$-far from satisfying $\P$ if applying a sequence of at most $\epsilon m$ edge-deletions or  edge-additions to $G$ cannot result in a graph satisfying $\P$. In the context of property testing, graphs are usually assumed to be stored using an adjacency list, and a centralized algorithm has the ability to probe nodes, via queries of the form $(i,j)$ where $i\in\{1,\dots,n\}$, and $j\geq 0$. The answer to a query $(i,0)$ is the degree of node~$i$, while the answer to a query $(i,j)$ with $j>0$ is the identity of the $j$th neighbor of node~$i$. After a small number of queries, the algorithm must output either accept or reject.  An algorithm $\seq$ is a testing algorithm for $\P$ if and only if, for every input graph $G$, 
\[
\left \{ \begin{array}{l}
\mbox{$G$ satisfies $\P \implies \Pr[\seq\; \mbox{accepts $G$}]\geq \frac23$;}  \\
\mbox{$G$ is $\epsilon$-far from satisfying $\P \implies \Pr[\seq\; \mbox{rejects $G$}]\geq \frac23$.}
\end{array} \right.
\]
(An algorithm is 1-sided if it systematically accepts every graph satisfying $\P$). The challenge in this context is to design testing algorithms performing as few queries as possible. 

In the context of \emph{distributed} property testing~\cite{CensorHillelFS16}, the challenge is not the number of queries (as all nodes perform their own queries in parallel), but the lack of global perspective on the input graph. The graph models a network. Every node of the network is a processor, and every processor can exchange messages with all processors corresponding to its neighboring nodes in the graph. After a certain number of rounds of computation, every node must output accept or reject.  A distributed algorithm $\dist$ is a distributed testing algorithm for $\P$ if and only if, for any graph $G$ modeling the actual network, 
\[
\left \{ \begin{array}{l}
\mbox{$G$ satisfies $\P \implies \Pr[\dist\; \mbox{accepts  $G$ in all nodes}]\geq \frac23$;}  \\
\mbox{$G$ is $\epsilon$-far from satisfying $\P \implies \Pr[\dist\; \mbox{rejects $G$ in at least one node}]\geq \frac23$.}
\end{array} \right.
\]

The challenge is to use as few resources of the network as possible. In particular, it is desirable that every processor could take its decision (accept or reject) without requiring data from far away processors in the network, and that processors exchange messages of size respecting the inherent bounds imposed by the limited bandwidth of the links. These two constraints are well captured by the  \Congest\ model. This model is a classical model for distributed computation~\cite{Peleg}. Processors are given distinct identities, that are supposed to have values in $[1,n^c]$ in $n$-node networks, for some constant $c\geq 1$. All processors start at the same time, and then proceed in synchronous rounds. At each round, every processor can send and receive messages to/from its neighbors, and perform some individual computation. All messages must be of size at most $O(\log n)$ bits. So, in particular, every message can include at most a constant number of processor identities. As a consequence, while every node can gather the identities of all its neighbors in just one round, a node with large degree that is aware of all the identities of its neighbors may not be able to send them all simultaneously to a given  neighbor. The latter observation enforces strong constraints on distributed testing algorithms in the  \Congest\ model. For instance, while the \Local\ model allows every node to gather its $t$-neighborhood in $t$ rounds, even just gathering the 2-neighborhood may require $\Omega(n)$ rounds in the  \Congest\ model (e.g., in the Lollipop graph). As a consequence, detecting the presence of even a small given pattern in a graph efficiently is not necessarily an easy task. 

The presence or absence of a certain given pattern (typically, paths, cycles or cliques of a given size) as a subgraph, induced subgraph, or minor, has a significant impact on graph properties, and/or on the ability to design efficient algorithms for hard problems. This paper investigates the existence of efficient distributed testing algorithms for the $H$-freeness property, depending on the given graph~$H$. Recall that, given a graph $H$, a graph $G$ is \emph{$H$-free} if and only if $H$ is not isomorphic to a subgraph of $G$, where $H$ is a subgraph of a graph $G$ if $V(H)\subseteq V(G)$, and $E(H)\subseteq E(G)$.  Recently, Censor-Hillel et al.~\cite{CensorHillelFS16} established a series of results regarding distributed testing of different graph properties, including bipartiteness and triangle-freeness. A triangle-free graph is a $K_3$-free graph or, equivalently, a $C_3$-free graph, where $K_k$ and $C_k$ respectively denote the clique and the cycle on $k$ vertices. The algorithm in~\cite{CensorHillelFS16} for testing bipartiteness (of bounded degree networks) requires $O(\log n)$ rounds in the \Congest\ model, and the authors conjecture that this is optimal. However, quite interestingly, the algorithm for testing triangle-freeness requires only a constant number of rounds, $O(1/\epsilon^2)$, i.e., it depends only on the (fixed) parameter $\epsilon$ quantifying the $\epsilon$-far relaxation. 

In this paper, we investigate the following question: what are the (connected) graphs $H$ for which testing $H$-freeness can be done in a constant number of rounds in the \Congest\ model?

\subsection{Our results} 

We show that, for every connected 4-node  graph $H$, testing whether a graph is $H$-free can  be done in $O(1/\epsilon^2)$ rounds. Hence, in particular, testing whether a graph is $K_4$-free, and  testing whether a graph is $C_4$-free can be done in a constant number of rounds. Our algorithm is generic in the sense that, for all 4-node  graphs $H$, the global communication structure of the algorithm is the same, with only a variant in the final decision for accepting or rejecting, which of course depends on $H$. 

In fact, we identify two different natural generic types of testing algorithms for $H$-freeness. We call the first type DFS tester, and our algorithm for testing the presence of 4-node patterns is actually the DFS tester. Such an algorithm applies to Hamiltonian graphs $H$, i.e., graphs $H$ containing a simple path spanning all its vertices. It performs in $|H|-1$ rounds.  Recall that, for a node set $A$, $G[A]$ denotes the subgraph of $G$ induced by $A$. At each round $t$ of the DFS tester, at every node $u$, and for each of its incident edges $e$, node $u$ pushes a graph $G[A_t]$ (where, initially, $A_0$ is just the graph formed by $u$ alone). The graph $G[A_t]$ is chosen u.a.r. among the  sets of graphs received from the neighbors at the previous round. More specifically, upon reception of every graph $G[A_{t-1}]$ at round $t-1$, node $u$ forms a graph $G[A_{t-1}\cup \{u\}]$, and the graph $G[A_t]$ pushed by $u$ along $e$ at round $t$ is chosen u.a.r. among the collection of graphs $G[A_{t-1}\cup \{u\}]$ currently held by~$u$. 

We call BFS tester the second type of generic testing algorithm for $H$-freeness. The algorithm in~\cite{CensorHillelFS16} for triangle-freeness is a simplified variant of the BFS tester, and we prove that the BFS tester can test $K_4$-freeness in  $O(1/\epsilon^2)$ rounds. Instead of guessing a path spanning $H$ (which may actually not exist for $H$ large), the BFS tester aims at directly guessing all neighbors in $H$ simultaneously. The algorithm performs in $D-1$ rounds, where $D$ is the diameter of $H$. For the sake of simplicity, let us assume that $H$ is $d$-regular. At each round, every node $u$ forms groups of $d$ neighboring nodes.  These groups may overlap, but a neighbor should not participate to more than a constant number of groups. Then for every  edge $e=\{u,v\}$ incident to $u$, node $u$ pushes all partial graphs of the form $G[A_{t-1}\cup \{v_1,\dots,v_t\}]$ where $v\in\{v_1,\dots,v_t\}$,  $A_{t-1}$ is chosen u.a.r. among the graphs received at the previous round, and the $v_i$'s will be in charge of checking the presence of edges between them and the other $v_j$'s. 

We prove that neither the DFS tester, nor the BFS tester can test $K_k$-freeness in a constant number of rounds for $k\geq 5$, and that the same holds for testing $C_k$-freeness (with the exception of a finite number of small values of $k$). This shows that testing  $K_k$-freeness or $C_k$-freeness for $k\geq 5$ in a constant number of rounds requires, if at all possible, to design algorithms which explore $G$ from each node in a way far more intricate than just parallel DFSs or parallel BFSs. Our impossibility results, although restricted to DFS and BFS testers, might be hints that testing $K_k$-freeness and $C_k$-freeness for $k\geq 5$ in $n$-node networks does require to perform a number of rounds which grows with the size~$n$ of the network. 

\subsection{Related work} 

The \Congest\ model  has become a standard model in distributed computation~\cite{Peleg}. Most of the lower bounds in the  \Congest\ model  have been obtained using reduction to  communication complexity problems~\cite{elkin2006unconditional,lotker2006distributed,SarmaHKKNPPW12}.  The so-called \emph{congested clique} model is the \Congest\ model in the complete graph $K_n$~\cite{drucker2014,lenzen2013,lotker03,patt2011}. There are extremely fast protocols for 
solving different types of graph problems in the congested clique, including finding ruling sets~\cite{hegeman2014b}, constructing minimum spanning trees~\cite{Hegeman2015}, and, closely related to our work, detecting small subgraphs~\cite{censor2015algebraic,dolev2012tri}. 

The distributed property testing framework in the  \Congest\ model was recently introduced in the aforementioned paper~\cite{CensorHillelFS16}, inspired from classical property testing~\cite{Goldreich2010,goldreich1998property}. Distributed property testing relaxes classical distributed decision~\cite{FraigniaudGKS13,fraigniaud2011local,GoosS11}, typically designed for the \textsf{LOCAL} model, by ignoring illegal instances which are less than $\epsilon$-far from satisfying the considered property. Without such a relaxation, distributed decision in the  \Congest\ model requires non-constant number of rounds~\cite{SarmaHKKNPPW12}. Other variants of local decision in the \textsf{LOCAL} model have been studied in~\cite{arfaoui2014distributedly,arfaoui2013local}, where each process can output an arbitrary value (beyond just a single bit --- accept or reject), and~\cite{becker2015allowing,kari2014solving}, where nodes are bounded to perform a single round before to output at most $O(\log n)$ bits. Distributed decision has been also considered in other distributed environments, such as shared memory~\cite{FraigniaudRT13}, with application to runtime verification~\cite{FraigniaudRT14},  and networks of finite automata~\cite{Reiter15}. 

\section{Detecting Small Graphs Using a DFS Approach}

In this section, we establish our main positive result, which implies, in particular, that testing $C_4$-freeness and testing $K_4$-freeness can be done in constant time in the  \Congest\ model.

\begin{theorem}\label{th:p4f}
Let $H$ be a graph on four vertices, containing a $P_4$ (a path on four vertices) as a subgraph. There is a 1-sided error distributed property-testing algorithm for $H$-freeness, performing in a constant number of rounds  in the \Congest\ model.
\end{theorem}

\begin{proof}
We show the theorem by exhibiting a generic distributed testing protocol for testing $H$-freeness, that applies to any graph $H$ on four vertices containing a $P_4$ as a subgraph. The core of this algorithm is presented as Algorithm~\ref{al:Pk}.  Note that the test $H\subseteq G[u,u',v',w']$ at step~4 of Algorithm~\ref{al:Pk} can be performed  thanks to the bit $b$ that tells about the only edge that node~$u$ is not directly aware of. Algorithm~\ref{al:Pk} performs in just two rounds (if we omit the round used to acquire the identities of the neighbors), and that a single $O(\log n)$-bit message is sent through every edge at each round. Clearly, if $G$ is $H$-free, then all nodes accept. 

\begin{algorithm}[htb]
\caption{Testing $H$-freeness for 4-node Hamiltonian $H$. Instructions for node $u$. }
\label{al:Pk}
Send $\id(u)$ to all neighbors; \\
For every neighbor $v$, choose a received identity $\id(w)$ u.a.r., and send $(\id(w),\id(u))$ to~$v$;\\ 
For every neighbor $v$, choose a received pair $(\id(w'),\id(u'))$ u.a.r., and send $(\id(w'),\id(u'),\id(u),b)$ to~$v$, where $b=1$ if $w'$ is a neighbor of $u$, and $b=0$ otherwise;\\
For every received 4-tuple $(\id(w'),\id(v'),\id(u'),b)$, check whether $H\subseteq G[u,u',v',w']$; \\
If $H\subseteq G[u,u',v',w']$ for one such 4-tuple then  reject else  accept. 
\end{algorithm}

In order to analyze the efficiency of Algorithm~\ref{al:Pk} in case $G$ is $\epsilon$-far from being $H$-free, let us consider a subgraph $G[\{u_1, u_2, u_3,u_4\}]$ of $G$ containing $H$, such that $(u_1, u_2, u_3,u_4)$ is a $P_4$ spanning $H$.  Let $\cal{E}$ be the event ``at step~2, vertex $u_2$ sends $(\id(u_1),\id(u_2))$ to its neighbor~$u_3$''. We have $\Pr[{\cal E}]=1/d(u_2)$. Similarly, let $\cal{E'}$ be the event ``at step~3, vertex $u_3$ sends $(\id(u_1),\id(u_2),\id(u_3))$ to its neighbor~$u_4$''.  We have $\Pr[ {\cal E'} | {\cal E}]=1/d(u_3)$. Since 
$
\Pr[ {\cal E} \wedge {\cal E'}]= \Pr[ {\cal E'} | {\cal E}]\cdot \Pr[{\cal E}], 
$
 it follows that 
\begin{equation}\label{eq:proba-cycle}
\Pr[H\; \mbox{is detected by $u_4$ while performing Algorithm~\ref{al:Pk}}]\geq \frac{1}{d(u_2)d(u_3)}. 
\end{equation}
Note that the events ${\cal E}$ and ${\cal E'}$ only depend on the choices made by $u_2$ for the edge $\{u_2,u_3\}$ and by $u_3$ for the edge $\{u_3,u_4\}$, in steps~3 and~4 of Algorithm~\ref{al:Pk}, respectively. Since these choices are performed  independently at all nodes, it follows that if $H_1$ and $H_2$ are edge-disjoint copies of $H$ in $G$, then the events ${\cal E}_1$  and ${\cal E}_2$ associated to them are independent, as well as the events ${\cal E'}_1$  and ${\cal E'}_2$

The following result will be used throughout the paper, so we state it as a lemma for further references. 

\begin{lemma}\label{lem:edge-disjoint}
Let $G$ be $\epsilon$-far from being $H$-free. Then $G$ contains at least $\epsilon m /|E(H)|$ edge-disjoint copies of $H$. 
\end{lemma}

To establish Lemma~\ref{lem:edge-disjoint}, let $S=\{e_1,\dots,e_k\}$ be a smallest set of edges whose removal from $G$ results in an  $H$-free graph. We have $k\geq \epsilon m$. Let us then remove these edges from $G$ according to the following process. The edges are removed in arbitrary order. Each time an edge $e$ is removed from $S$, we select an arbitrary copy $H_e$ of $H$ containing $e$, we remove all the edges of $H_e$ from $G$, and we reset $S$ as $S\setminus E(H_e)$. We proceed as such until we have exhausted all the edges of $S$. Note that each time we pick an edge $e\in S$, there always exists a copy $H_e$ of $H$ containing $e$. Indeed, otherwise, $S\setminus\{e\}$ would also be a set  whose removal from $G$ results in an  $H$-free graph, contradicting the minimality of $|S|$. After at most $k$ such removals, we get a graph that is $H$-free, and, by construction,  for every two edges $e,e'\in S$, we have that $H_e$ and $H_{e'}$ are edge-disjoint.  Every step of this process removes at most $|E(H)|$ edges from~$S$, hence the process performs at least $\epsilon m /|E(H)|$ steps before exhausting all edges in~$S$. Lemma~\ref{lem:edge-disjoint} follows.  

\medskip

Let us now define an edge $\{u,v\}$ as \emph{important} if it is the middle-edge of a $P_4$ in one of the $\epsilon m/|E(H)|$ edge-disjoint copies of $H$ constructed in the proof of Lemma~\ref{lem:edge-disjoint}. We denote by $I(G)$ the set of all important edges. Let $N_0$ be the random variable counting the number of distinct copies of $H$ that are detected by Algorithm~\ref{al:Pk}. As a direct consequence of Eq~\eqref{eq:proba-cycle}, we get that 
\[
\E(N_0) \geq \sum_{\{u,v\} \in I(G)} \frac{1}{d(u)d(v)}.
\]
Define an edge $\{u,v\}$ of $G$ as \emph{good} if $d(u)d(v) \leq 4 m |E(H)| / \epsilon$, and let $g(G)$ denote the set of good edges. Note that if there exists a constant $\gamma>0$ such that $|I(G)\cap g(G)|\geq \gamma m$, then the expected number of copies of $H$ detected during a phase is 
\begin{equation}\label{eq:expwithgamma}
\E(N_0) \geq \sum_{\{u,v\} \in I(G)\cap g(G)} \frac{1}{d(u)d(v)} \geq \gamma m \, \frac{1}{4 m |E(H)|/\epsilon} = \frac{\gamma \epsilon}{4 |E(H)|}.
\end{equation}
We now show that the number of edges that are both important and good is indeed at least a fraction $\gamma$ of the egdes, for some constant $\gamma>0$. We first show that  $G$ has at least $(1-\frac{3}{4|E(H)|} \epsilon) m$ good edges. Recall that $\sum_{u \in V(G)} d(u) = 2m$, and define $N(u)$ as the set of all neighbors of node~$u$. We have 
\[
\sum_{\{u,v\} \in E(G)} d(u)d(v) = \frac12  \sum_{u \in V(G)} d(u) \sum_{v\in N(u)} d(v) \leq \sum_{u\in V(G)}d(u) \; m \leq 2 m^2.
\]
Thus $G$ must have at least $(1-\frac{3}{4|E(H)|}\epsilon) m$ good edges, since otherwise
\[
\sum_{\{u,v\} \in E(G)} d(u)d(v) \geq \sum_{\{u,v\} \in E(G)\setminus g(G)} d(u)d(v) > \frac{3}{4|E(H)|} \epsilon m  \frac{4m|E(H)|}{\epsilon} = 3 m^2,
\]
contradicting the aforementioned $2 m^2$ upper bound. Thus, $G$ has at least $(1 - \frac{3}{4|E(H)|} \epsilon)m$ good edges. On the other hand, since the number of important edges is at least the number of edge-disjoint copies of $H$ in $G$, there are at least $\epsilon m /|E(H)|$ important edges. It follows that the number of edges that are both important and good  is at least $\frac{\epsilon}{4|E(H)|}m$. 
Therefore, by Eq.~\eqref{eq:expwithgamma}, we get that 
\[
\E(N_0) \geq \left (\frac{\epsilon}{4 |E(H)|} \right)^2.
\]
All the above calculations were made on the $\epsilon m/|E(H)|$ edge-disjoint copies of $H$ constructed in the proof of Lemma~\ref{lem:edge-disjoint}. Therefore, if $X_i^{(0)}$ denotes the random variable satisfying $X_i^{(0)}=1$ if the $i$th copy $H$ is detected, and $X_i^{(0)}=0$ otherwise, then we have  $N_0 = \sum_{i=1}^{\epsilon m/|E(H)|}X_i^{(0)}$, and the variables $X_i^{(0)}$, $i=1,\dots,\epsilon m/|E(H)|$, are mutually independent.  Let 
\[
T= 8 \ln 3 \; \left (\frac{4 |E(H)|}{\epsilon} \right)^2.
\]
By repeating the algorithm $T$ times, and defining $N=\sum_{t=0}^{T-1}N_t$ where $N_t$ denotes the number of copies of $H$ detected at the $t$th independent repetition, we get 
\[
\E(N)\geq 8 \ln 3. 
\]
In fact, we also have $N=\sum_{t=0}^{T-1}\sum_{i=1}^{\epsilon m/|E(H)|}X_i^{(t)}$ where $X_i^{(t)}=1$ if the $i$th copy $H$ is detected at the $t$th iteration of the algorithm, and $X_i^{(t)}=0$ otherwise. All these variables are mutually independent, as there is mutual independence within each iteration, and all iterations are performed independently. Therefore, Chernoff bound applies (see Theorem~4.5 in \cite{Mitzenmacher}), and so, for every $ 0 < \delta < 1$, we have 
\[\Pr[N \leq (1-\delta) \E[N]] \leq e^{-\delta^2 \E[N] /2}.\] 
By taking $\delta =  \frac12$ we get  
\[
\Pr[N \leq 4 \ln 3] \leq \frac13.
\]
Therefore, a copy of $H$ is detected with probability at least $\frac{2}{3}$, as desired, which completes the proof of Theorem~\ref{th:p4f}. 
\end{proof}

\paragraph{Remark.}  Theorem~\ref{th:p4f} requires the existence of a $P_4$ in the graph $H$. All connected graphs $H$ of four vertices satisfy this condition, with the only exception of the \emph{star} $K_{1,3}$ (a.k.a.~the \emph{claw}). Nevertheless, detecting whether a graph $G$ has a star $K_{1,3}$ as a subgraph is trivial (every node rejects whenever its degree at least three).

\section{Limits  of the DFS Approach}

Algorithm~\ref{al:Pk} can be extended in a natural way to any  $k$-node graph $H$ containing a Hamiltonian path, as follows. Given a graph $G=(V,E)$, and $A\subseteq V$, let $G[A]=(A,E_A)$ denote the subgraph of $G$ induced by $A$ (i.e., for every two nodes $u$ and $v$ in $A$, $\{u,v\}\in E_A \iff \{u,v\}\in E$). The extension of Algorithm~\ref{al:Pk} to Hamiltonian graphs $H$ is depicted in Algorithm~\ref{al:PkGen}.  At the first round, every vertex $u$ sends its identifier to its neighbors, and composes the $d(u)$ graphs formed by the edge $\{u,v\}$, one for every neighbor $v$. Then, during the $k-2$ following rounds, every node $u$ forwards through each of its edges one of the graphs formed a the previous round. 

\begin{algorithm}[htb]
\caption{Testing $H$-freeness: Hamiltonian $H$, $|V(H)|=k$. Instructions for node~$u$. }
\label{al:PkGen}
\SetKwFor{RepeatTimes}{repeat}{times}{}
send the 1-node graph $G[u]$ to every neighbor $v$; \\
form the graph $G[\{u,v\}]$ for every neighbor $v$; \\
\RepeatTimes {$k-2$}
{
	\For {every neighbor $v$}
	{
		choose a graph $G[A]$ u.a.r. among those formed during the previous round; \\
	 	send $G[A]$ to $v$; \\
		receive the graph $G[A']$ from $v$;\\
	 	form the graph $G[A'\cup \{u\}]$;\\
	}
}
if $H \subseteq G[A]$ for one of the graphs formed at the last round then reject else accept.
\end{algorithm}

Let $(u_1,u_2,\dots,u_k)$ be a simple path in $G$, and assume that $G[\{u_1,u_2,\dots,u_k\}]$ contains $H$. If, at each round $i$, $2 \leq i <k$, vertex $u_i$ sends to $u_{i+1}$ the graph $G[\{u_1,\dots,u_i\}]$, then, when the repeat-loop completes, vertex $u_k$ will test precisely the graph $G[\{u_1,u_2,\dots,u_k\}]$, and thus $H$ will be detected by the algorithm. 

Theorem~\ref{th:p4f} states that Algorithm~\ref{al:PkGen} works fine for 4-node graphs $H$. We show that, $k=4$ is precisely the limit of detection for graphs that are $\epsilon$-far from being $H$-free, even for  the cliques and the cycles.

\begin{theorem}\label{th:neg}
Let $H = K_k$ for arbitrary $k\geq 5$, or $H = C_k$ for arbitrary odd $k\geq 5$. There exists a graph $G$ that is $\epsilon$-far from being $H$-free in which any constant number of repetitions of Algorithm~\ref{al:PkGen} fails to detect $H$, with probability at least $1-o(1)$. 
\end{theorem}

For the purpose of proving Theorem~\ref{th:neg}, we use the following combinatorial result, which extends Lemma~7 of~\cite{AKKR08}, where the corresponding claim was proved for $k = 3$, with a similar proof. The bound on $p'$ is not even nearly optimal in Lemma~\ref{le:largeX} below, but it is good enough for our purpose\footnote{The interested reader can consult \cite{Schoen2014} for the state-of-the-art on such combinatorial constructions,  in particular constructions for $p' \geq  p^{1 - c /\sqrt{\log p}}$, for a constant $c$ depending on $k$.}.

\begin{lemma}\label{le:largeX}
Let $k$ be a constant. For any sufficiently large $p$, there exists a set $X \subset \{0,\dots, p-1\}$ of size $p'  \geq \ppr$ such that, for any $k$ elements $x_1,x_2,\dots, x_k$ of $X$,
$$ \sum_{i=1}^{k-1} x_i \equiv (k-1) x_k \pmod p \;\; \implies \;\;  x_1 = \dots = x_{k-1} = x_k.$$ 
\end{lemma}


\begin{proof}
Let $b = \lfloor \log p \rfloor$ and $a = \left\lfloor \frac{\log p}{\log \log p} \right\rfloor$. Take $p$ sufficiently large so that $a < b/k$ is satisfied. $X$~is a set of integers encoded in base $b$, on $a$ $b$-ary digits, such that the digits of each $x \in X$ are a permutation of $\{0,1,\dots, a-1\}$. More formally, for any permutation $\pi$ over $\{0, \dots, a-1\}$, let $N_\pi = \sum_{i=0}^{a-1} \pi(i) b^i$. Then, let us set
$
X = \{N_{\pi} \mid \pi \text{~is a permutation of~} \{0,\dots, a-1\}\}.
$
Observe that different permutations $\pi$ and $\pi'$ yield different numbers $N_{\pi}$ and $N_{\pi'}$ because these numbers have different digits in base $b$.  Hence $X$ has $p' = a!$ elements. Using the inequality $z! > (z/e)^z$ as in \cite{AKKR08} (Lemma~7), we get that  $a! \geq  \ppr$, as desired. 

Now, for any $x \in X$, we have $x \leq p/k$. Indeed,  $x < a\cdot b^{a-1} \leq \frac{1}{k} b^a$, and $b^a \leq (\log p)^\frac{\log p}{\log \log p} = p$. Consequently, the modulo in the statement of the Lemma becomes irrelevant, and we will simply consider integer sums. Let $x_1,\dots,x_{k-1},x_k$ in $X$, such that  $\sum_{l=1}^{k-1} x_i = (k-1) x_k$. Viewing the $x_i$'s as integers in base~$b$, and having in mind that all digits are smaller than $b/k$, we get that the equality must hold coordinate-wise.  For every $1 \leq l \leq k$, let $\pi_l$ be the permutation such that $x_l = N_{\pi_l}$. For every $i \in {0, \dots, a-1}$, we have
\[
\sum_{l=1}^{k-1} \pi_l(i) = (k-1) \pi_k(i).
\]
By the Cauchy-Schwarz inequality applied to vectors $(\pi_1(i),\dots,\pi_{k-1}(i))$ and $(1,\dots, 1)$, for every $i \in \{0, \dots, a-1\}$, we also have
\[
\sum_{l=1}^{k-1} \left(\pi_l(i)\right)^2 \geq (k-1) \left(\pi_k(i)\right)^2.
\]
Moreover equality holds if and only if  $\pi_1(i) = \dots =  \pi_{k-1}(i) = \pi_k(i)$. By summing up the $a$ inequalities induced by the $a$ coordinates, observe that both sides sum to exactly $(k-1)\sum_{i=0}^{a-1} i^2$. Therefore, for every $i$, the Cauchy-Schwarz inequality is actually an equality, implying that the $i$th digit is identical in all the $k$ integers $x_1, \dots, x_k$. As a consequence, $x_1 = \dots = x_{k-1} = x_k$, which completes the proof. 
\end{proof}

\paragraph{Proof of Theorem~\ref{th:neg}.}

Assume that $G[\{u_1,u_2,\dots,u_k\}]$ contains $H$, where $(u_1,u_2,\dots,u_k)$ is a path of $G$. For  $2 \leq i \leq k-1$, let us consider the event ``at round $i$, vertex $u_i$ sends  the graph $G[u_1,\dots,u_i\}]$ to $u_{i+1}$''. Observe that this event happens with probability $\frac{1}{d(u_i)}$ because $u_i$ choses which subgraph to send uniformly at random among the $d(u_i)$ constructed subgraphs. With the same arguments as the ones used to establish Eq.~\eqref{eq:proba-cycle}, we get
\begin{equation}\label{eq:proba-cycle-gen}
\Pr[H\; \mbox{is detected along the path $(u_1,\dots,u_k)$}] = \frac{1}{d(u_2)d(u_3)\dots d(u_{k-1})}. 
\end{equation}
We construct families of graphs which will allow us to show, based on that latter equality, that the probability to detect a copy of $H$ vanishes with the size of the input graph $G$. We actually use a variant of the so-called \emph{Behrend graphs} (see, e.g.,~\cite{AKKR08,Behrend}), and we construct graph families indexed by $k$, and by a parameter $p$, that we denote by $BC(k,p)$ for the case of cycles, and by $BK(k,p)$ for the case of cliques. We prove that these graphs are $\epsilon$-far from being $H$-free, while the probability that Algorithm~\ref{al:PkGen} detects a copy of $H$ in these graphs goes to $0$.

\medskip

Let us begin with the case of testing cycles. Let $p$ be a large prime number, and let $X$ be a subset of $ \{0,\dots, p-1\}$ of size $p' \geq \ppr$, where $p'$ is as defined  in Lemma~\ref{le:largeX}. Graph $BC(k,p)$ is then constructed as follows. The vertex set $V$ is the disjoint union of an odd number $k$ of sets, $V^0, V^2, \dots V^{k-1}$, of $p$ elements each. For every $l$, $0 \leq l \leq k-1$, let $u^l_i$, $i=0,\dots,p-1$ be the nodes in $V^l$ so that 
\[
V^l = \{u^l_i \mid i \in \{0, \dots, p-1\}\}.
\] 
For every $i \in \{0, \dots, p-1\}$ and every $x \in X$, edges in $BC(k,p)$ form a cycle 
\[
C_{i,x} = (u^0_i,  \dots, u^l_{i+lx}, u^{l+1}_{i+(l+1)x},\dots,u^{k-1}_{i+(k-1)x}),
\]
where the indices are taken modulo $p$. The cycles $C_{i,x}$ form a set of $pp'$ edge-disjoint copies of $C_k$ in $BC(k,p)$. Indeed, for any two distinct pairs $(i,x) \neq (i',x')$, the cycles $C_{i,x}$ and $C_{i',x'}$ are edge-disjoint. Otherwise there exists a common edge $e$ between the two cycles. It can be either between two consecutive layers $V^l$ and $V^{l+1}$, or between $V^0$ and $V^{k-1}$. There are two cases. If $e=\{u^l_y,u^{l+1}_z\}$ we must have $y = i+lx = i'+lx'$ and $z = i+(l+1)x = i'+(l+1)x'$, where equalities are taken modulo $p$, and, as a consequence, $(i,x) = (i',x')$. If $e = \{u^0_y,u^{k-1}_z\}$ then we have $y=i=i'$ and $z=i+(k-1)x' = i'+(k-1)x'$, and, since $p$ is prime, we also conclude that $(i,x) = (i',x')$. 

We now show that any $k$-cycle has exactly one vertex in each set $V^l$, for $0 \leq l < k-1$. For this purpose, we focus on the parity of the layers formed by consecutive vertices of a cycle. The ``short'' edges (i.e., ones between consecutive layers) change the parity of the layer, and hence every cycle must include ``long'' edges (i.e., ones between layers $0$ and $k-1$). However, long edges do not change the parity of the layer. Therefore, every cycle contains an odd number of long edges, and an even number of short edges. For this to occur, the only possibility is that the cycle contains a vertex from each layer.

Next, we show that any cycle of $k$ vertices in $BC(k,p)$ must be of the form $C_{i,x}$ for some pair $(i,x)$. Let $C = (u^0_{y_0}, \dots u^l_{y_l}, u^{l+1}_{y_{l+1}}, \dots, u^{k-1}_{y_{k-1}})$ be a cycle in $BC(k,p)$. For every $l=1, \dots, k-1$, let $x_l = y_l - y_{l-1} \bmod p$. We have  $x_l \in X$ because the edge $\{u^{l-1}_{y_{l-1}},u^l_{y_l}\}$ is in some cycle $C_{i,x_l}$. Also set $x_k \in X$ such that the edge $\{u^0_{y_0},u^{k-1}_{y_{k-1}}\}$ is in the cycle $C_{y_0,x_k}$. In particular we must have that $y_{k-1} = y_0 + (k-1)x_k$. It follows  that $y_{k-1} = y_0 + (k-1)x_k = y_0+x_1 + x_2+ \dots + x_{k-1}$. By Lemma~\ref{le:largeX}, we must have $x_1 = \dots = x_{k-1} = x_k$, so $C$ is of the form $C_{i,x}$.

It follows from the above that $BC(k,p)$ has exactly $pp'$ edge-disjoint cycles of $k$ vertices. Since $BC(k,p)$ has $n = kp$ vertices, and $m = kpp'$ edges, $BC(k,p)$ is $\epsilon$-far from being $C_k$-free, for $\epsilon = \frac{1}{k}$. Also, each vertex of $BC(k,p)$ is of degree $2p'$ because each vertex belongs to $p'$ edge-disjoint cycles. 

Let us now consider an execution of Algorithm~\ref{al:PkGen} for input $BC(k,p)$. As $BC(k,p)$ is regular of degree $d = 2p'$, this execution has probability at most $\frac{2k}{d^{k-2}}$ to detect any given cycle $C$ of $k$ vertices. Indeed, $C$ must be detected along one of the paths formed by its vertices in graph $BC(k,p)$, there are at most $2k$ such paths in $C$ (because all $C_k$'s in $BC(k,p)$ are induced subgraphs, and paths are oriented), and, by Eq.~\eqref{eq:proba-cycle-gen}, the probability of detecting the cycle along one of its paths is $\frac{1}{d^{k-2}}$. Therefore,  applying the union bound, the probability of detecting a given cycle $C$ is at most $\frac{2k}{d^{k-2}}$.

Since there are $pp'$ edge-disjoint cycles, the expected number of cycles detected in one execution of Algorithm~\ref{al:PkGen} is at most $\frac{2k pp'}{(2p')^{k-2}}$. It follows that the expected number of cycles detected by repeating the algorithm $T$ times is at most $\frac{2k pT}{2(2p')^{k-3}}$. Consequently, the probability that the algorithm detects a cycle is at most $\frac{2k pT}{2(2p')^{k-3}}$. Plugging in the fact that, by Lemma~\ref{le:largeX}, $p' = p^{1-o(1)}$, we conclude that, for any constant $T$, the probability that $T$ repetitions of Algorithm~\ref{al:PkGen} detect a cycle goes to $0$ when $p$ goes to $\infty$, as claimed. 

\medskip

The case of the complete graph is treated similarly. Graphs $BK(k,p)$ are constructed in a way similar to $BC(k,p)$, as $k$-partite graphs with $p$ vertices in each partition (in particular,  $BK(k,p)$  also has $n = kp$ vertices). The difference with $BC(k,p)$ is that, for each pair $(i,x) \in \{0,\dots,p-1\}\times X$, we do not add a cycle, but a complete graph $K_{i,x}$ on the vertex set $\{u^0_i,  \dots, u^l_{i+lx}, u^{l+1}_{i+(l+1)x},\dots,$ $u^{k-1}_{i+(k-1)x}\}$.  By  the same arguments as for $BC(k,p)$, $BK(k,p)$ contains contains exactly $pp'$ edge-disjoint copies of $K_k$ (namely $K_{i,x}$, for each pair $(i,x)$). This fact holds even for even values of $k$, because any $k$-clique must have a vertex in each layer, no matter the parity of $k$. Thus, in particular $BK(k,p)$  has $m = \binom{k}{2}pp'$ edges, and every vertex is of degree $d= (k-1)p'$. The graph $BK(k,p)$ is $\epsilon$-far from being $K_k$-free, for $\epsilon = \frac{2}{k(k-1)}$. The probability that Algorithm~\ref{al:PkGen} detects a given copy of $K_k$ is at most $\frac{k!}{d^{k-2}}$. Indeed, a given $K_k$ has $k!$ (oriented) paths of length $k$, and, by Eq.~\eqref{eq:proba-cycle-gen}, the probability that the algorithm detects this copy along a given path is $\frac{1}{d^{k-2}}$. The expected number $\E[N]$ of $K_k$'s detected in $T$ runs of Algorithm~\ref{al:PkGen} is, as for $BC(k,p)$, at most  $\frac{k!\, Tpp'}{d^{k-2}} = \frac{k! \, pT}{(k-1)^{k-2}(p')^{k-3}}$. Therefore, since $\Pr[N\neq 0]\leq \E[N]$, we get that  the probability that the algorithm detects some $K_k$ goes to $0$ as $p$ goes to $\infty$. It follows that the algorithm fails to detect $K_k$, as claimed.
 \qed

\paragraph{Remark.} The proof that Algorithm~\ref{al:PkGen} fails to detect $C_k$ for odd $k\geq 5$, can be extended to all (odd or even) $k\geq 13$, as well as to $k=10$. The cases of $C_6$, $C_8$, are $C_{12}$ are open, although we strongly believe that Algorithm~\ref{al:PkGen} also fails to detect these cycles in some graphs.

\section{Detecting Small Graphs Using a BFS Approach}

We discuss here another very natural approach, extending the algorithm proposed by Censor-Hillel \text{et al.}~\cite{CensorHillelFS16} for testing triangle-freeness. In the protocol of~\cite{CensorHillelFS16}, each node $u$ samples two neighbors $v_1$ and $v_2$ uniformly at random, and asks them to check the presence of an edge between them. We generalize this protocol as follows. Assume that the objective is to test $H$-freeness, for a graph $H$ containing a universal vertex. Each node $u$ samples $d(u)$ sets $S_1,  \dots, S_{d(u)}$, of $|V(H)|-1$ neighbors each. For each  $i=1,\dots,d(u)$, node $u$ sends $S_i$ to all its neighbors in $S_i$, asking them to check the presence of edges between them, and collecting their answers. Based on these answers, node $u$ can  tell whether $G[S_i \cup \{u\}]$ contains $H$. We show that this very simple algorithm can be used for testing $K_4$-freeness.

\begin{theorem}\label{th:H4univ}
There is a 1-sided error distributed property-testing algorithm for $K_4$-freeness, performing in a constant number of rounds  in the \Congest\ model.
\end{theorem}

\begin{proof}
Again, we show the theorem by exhibiting a generic distributed testing protocol for testing $H$-freeness, that applies to any graph $H$ on four vertices with a universal vertex. The core of this algorithm is presented as Algorithm~\ref{al:K4gen} where all calculations on indices are performed modulo $d=d(u)$ at node $u$. This algorithm is presented for a graph $H$ with $k$ nodes. 

\begin{algorithm}[htb]
\caption{Testing $H$-freeness for $H$ with a universal vertex. Instructions for node $u$ of degree $d$. We let $k = |V(H)|$.}
\label{al:K4gen}
send $\id(u)$ to all neighbors; \\
index the $d$ neighbors $v_0,\dots,v_{d-1}$ in increasing order of their IDs; \\
pick a permutation $\pi\in \Sigma_d$ of $\{0,1,\dots, d-1\}$, u.a.r.; \\ \label{step3}
\For {\rm each $i\in\{0,\dots,d-1\}$ \label{step4}}
{Send $(\id(v_{\pi(i)}), \id(v_{\pi(i+1)}), \dots , \id(v_{\pi(i+k-2)}))$ to $v_{\pi(i)}, v_{\pi(i+1)}, \dots, v_{\pi(i+k-2)}$;  \label{step5} \\
}
 	\For {\rm each $i\in\{0,\dots,d-1\}$\label{step6}}
		{
		\For {\rm each of the $k-1$ tuples $(\id(w^{(1)}),\dots,\id(w^{(k-1)}))$ received from $v_i$} 
			{
			 Send $(b^{(1)},\dots,b^{(k-1)})$ to $v_i$ where  $b^{(j)}=1$ iff $u=w^{(j)}$ or $\{u,w^{(j)}\}\in E$;  \label{step8}\\
			 }
		}
If $\exists i\in \{0,\dots,d-1\}$ s.t. $H\subseteq G[u,v_{\pi(i)},\dots, v_{\pi(i+k-2)}]$ then  reject else  accept. 
\end{algorithm}

At Step~\ref{step3}, node $u$ picks a permutation $\pi$ u.a.r., in order to compose the $d(u)$ sets  $S_1,  \dots, S_{d(u)}$, which are sent in  parallel at Steps~\ref{step4}-\ref{step5}. At Step~\ref{step8}, every node $u$ considers separately each of the $k-1$ tuples of size $k-1$ received from each of its neighbors, checks the presence of edges between $u$ and each of the nodes in that tuple, and sends back the result to the neighbor from which it received the tuple. Finally, the tests $H\subseteq G[u,v_{\pi(i)},v_{\pi(i+1)},\dots, v_{\pi(i+k-2)}]$ performed at the last step is achieved thanks to the  $(k-1)$-tuple of bits received from each of the neighbors $v_{\pi(i)},v_{\pi(i+1)}, \dots, v_{\pi(i+k-2)}$, indicating the presence or absence of all the edges between these nodes. Note that exactly $2k-5$ IDs are actually sent through each edges at Steps~\ref{step4}-\ref{step5}, because of the permutation shifts. Similarly, $2k-5$ bits are sent through each edge at Steps~\ref{step6}-\ref{step8}. Therefore  Algorithm~\ref{al:K4gen} runs in a constant number of rounds in the \Congest\ model.

We now show that, for any given $G$ which is $\epsilon$-far from being $K_4$-free, Algorithm~\ref{al:K4gen} for $H=K_4$ rejects $G$ with constant probability. 

Let $K$ be a copy of $K_4$ containing a node $u$. Note that $d(u) \geq k-1$ (this condition will be implicitly used in the sequel). Observe that, if $u$ samples the vertices of $K \setminus \{u\}$ (Step~\ref{step5} of Algorithm~\ref{al:K4gen}), then $u$ will detect~$K$. Let us first prove that
\begin{equation}\label{eq:prK4}
\Pr[u \text{~detects~} K] =  d(u)/{d(u) \choose {k-1}}.
\end{equation}

Indeed, among the $d(u)!$ permutations $\pi$ of $\Sigma_{d(u)}$, the ones that sample $K \setminus \{u\}$ are those where the vertices of $K \setminus \{u\}$ appear consecutively as $v_{\pi(i)}, \dots, v_{\pi(i+d-1)}$, for some $i, 0\leq i \leq d(u)-1$ (recall that indices are taken modulo $d$ in Algorithm~\ref{al:K4gen}). We say that we \emph{insert}  a sequence $(y_0, \dots y_{k-2})$ into a sequence $(x_0, \dots, x_i, x_{i+1} \dots, x_{d(u)-k})$ starting from position $i \leq d(u) - k + 1$, for creating the sequence
\[
(x_0, \dots, x_{i-1}, y_0, \dots, y_{k-2}, x_i, \dots, x_{d(u)-k}).
\] 
If $i>d(u)-k+1$, the insertion of  $(y_0, \dots y_{k-2})$ into $(x_0, \dots, x_{d(u)-k})$ at position $i$ means the sequence 
\[
(y_{d(u)-i},\dots,y_{k-2},x_0, \dots, x_{d(u)-k},y_0,\dots,y_{d(u)-i-1}).
\]
The permutations that sample $K \setminus \{u\}$ are the permutations obtained from one of the $\left(d(u)-k+1\right)!$ permutations of the vertices of $N(u) \setminus K$ by inserting one of the $(k-1)!$ permutations of $K \setminus \{u\}$ starting from position $i$, $ 0\leq i \leq d(u)-1$, which proves Equation~\ref{eq:prK4}. 

\bigskip

Since $G$ is $\epsilon$-far from being $K_4$-free, it follows from Lemma~\ref{lem:edge-disjoint} that the graph $G$ contains a family of $\frac{\epsilon m}{6}$ edge-disjoint copies of $K_4$. For each vertex $u$, denote by $c(u)$ the number of copies in this family that contain $u$.  For $i=0,\dots,c(u)-1$, let $X_u^{(i)}$ be the random variable indicating whether $u$ detects the $i$th copy of $K_4$  incident to it, and denote by $X_u$ the random variable indicating whether $u$ detects at least one of these copies. We have:
\[
\Pr[X_u=1] = \Pr[ (X_u^{(0)}=1) \vee (X_u^{(1)}=1) \vee \dots \vee (X_u^{(c(u)-1)}=1)].
\]
By the inclusion-exclusion principle and Bonferroni's inequality~\cite{Bonferroni} (Exercise~1.7 in \cite{Mitzenmacher}), we get
\begin{equation}\label{eq:Bonf}
\Pr[X_u=1] \geq  \sum_{i=1}^{c(u)} \Pr[X_u^{(i)}=1 ] \;\; - \sum_{1 \leq i < j \leq c(u)} \Pr[X_u^{(i)}=1  \wedge X_u^{(j)}=1] .
\end{equation}
Now, for every $i$, $\Pr[X_u^{(i)}=1]$ is given by Equation~\ref{eq:prK4}. In particular,
\[
\Pr[X_u^{(i)}=1 ] = \frac{(k-1)!}{(d(u)-1)(d(u)-2)\dots(d(u)-k+2)},
\]
which, for $k=4$, gives
\begin{equation}\label{eq:Pr4}
\Pr[X_u^{(i)}=1 ] \geq \frac{6}{(d(u)-1)(d(u)-2)},
\end{equation}
and
\begin{equation}\label{eq:SumPr4i}
\sum_{i=1}^{c(u)}\Pr[X_u^{(i)}=1 ] \geq \frac{6 c(u)}{(d(u)-1)(d(u)-2)}.
\end{equation}
We now claim that, for any pair $i,j$ such that $0 \leq i < j \leq c(u)$, we have 
\[
\Pr[X_u^{(i)}=1  \wedge X_u^{(j)}=1] \leq \frac{d(u)^2((k-1)!)^2 (d(u)-2(k-1))!}{d(u)!}. 
\]
Indeed, let $K^i$ and $K^j$ be edge-disjoint copies of $K_4$  containing $u$ corresponding to the random variables $X_u^{(i)}$ and $X_u^{(j)}$. Among the  $d(u)!$ permutations of $N(u)$, those who will satisfy $X_u^{(i)} = X_u^{(j)} = 1$ are such that the elements of $K^i \setminus \{u\}$ and of  $K^j \setminus \{u\}$ appear consecutively (in the cyclic order). Such a permutation is obtained from one of the $(d(u)-2(k-1))!$ permutations of $N(u) \setminus (K^i \cup K^j)$, by inserting one of the $(k-1)!$ permutations of $K^i \setminus \{u\}$, and one of the  $(k-1)!$ permutations of $K^j \setminus \{u\}$, each at one of the $d(u)$ possible positions between $0$ and $d(u)-1$. The inequality holds automatically if $d(u) \geq 2(k-1)$, as otherwise there cannot be two edge-disjoint copies of $K_k$ containing $u$. In particular, for $k=4$, we obtain
\[
\Pr[X_u^{(i)}=1  \wedge X_u^{(j)}=1] \leq \frac{36 d(u)^2 (d(u)-6)!}{d(u)!} = \frac{36 d(u)}{(d(u)-1)(d(u)-2)\dots(d(u)-5)}.
\]
Therefore, 
\[
\sum_{1 \leq i < j \leq c(u)} \Pr[X_u^{(i)}=1  \wedge X_u^{(j)}=1] \leq {c(u) \choose 2}\frac{36 d(u)}{(d(u)-1)(d(u)-2)\dots(d(u)-5)}.
\]
For each $u$, we have $c(u) \leq \frac{d(u)}{3}$ because $c(u)$ counts edge-disjoint copies of $K_4$ containing $u$. So 
\[
{c(u) \choose 2} = \frac{c(u)(c(u)-1)}{2} \leq c(u) \cdot \frac{d(u)-3}{6},
\]
and
\begin{equation}\label{eq:SumPr4ij}
\sum_{0 \leq i < j \leq c(u)} \Pr[X_u^{(i)}=1  \wedge X_u^{(j)}=1] \leq c(u) \frac{6 d(u)}{(d(u)-1)(d(u)-2)(d(u)-4)(d(u)-5)}.
\end{equation}
By Eq.~\eqref{eq:Bonf}, and by subtracting Eq.~\eqref{eq:SumPr4i} and~\eqref{eq:SumPr4ij}, we obtain
\[
\Pr[X_u=1] \geq \frac{6c(u)}{(d(u)-1)(d(u)-2)}\left(1 - \frac{d(u)}{(d(u)-4)(d(u)-5)}\right) .
\]
If $d(u) \geq 8$, the the expression in the parenthesis is at least $1/3$. Therefore, if $d(u) \geq 8$ then 
\[
 \Pr[X_u=1] \geq \frac{2c(u)}{d(u)^2}.
\]
For vertices $u$ such that $d(u) <8$, observe that $c(u) \leq 2$. Therefore, thanks to Eq.~\eqref{eq:Pr4}, we deduce that the probability that $u$ detects a $K_4$ is at least $\frac{2c(u)}{d(u)^2}$. Consequently, for any vertex~$u$,
\begin{equation}\label{eq:prXuK4}
\Pr[X_u = 1] \geq \frac{2 c(u)}{d(u) ^2}.
\end{equation}
Let $Y$ denote the random variable counting the number of vertices $u$ such that $u$ detects one of the $\frac{\epsilon m}6$ edge-disjoint copies of $K_4$. Hence $Y = \sum_{u \in V(G)} X_u$ and, by Equation~\ref{eq:prXuK4}, 
\[
\E[Y] \geq  \sum_{u \in V(G)} \frac{2 c(u)}{d(u) ^2}
\]
Let us prove that this expectation is lower bounded by a positive constant. We say that the vertex $u$ is \emph{good} if $d(u) \leq \frac{24c(u)}{\epsilon}$. Otherwise, $u$ is \emph{bad}. The set of good vertices is denoted by $Good(G)$ and the set of bad vertices is denoted by $Bad(G)$. Observe that 
\[
\sum_{u \in Bad(G)} c(u) < \frac{\epsilon m}{12},
\]
because  
\[
2m = \sum_{u \in V(G)} d(u) \geq  \sum_{u \in Bad(G)} d(u) >  \sum_{u \in Bad(G)} \frac{24c(u)}{\epsilon}.
\]
Therefore, among the $\frac{\epsilon m}{6}$ edge-disjoint copies of $K_4$, at least half do not contain bad nodes, i.e., all their nodes are good. Hence, the number of edges between good vertices is at least $\frac{\epsilon m}{2}$.

\medskip

Now, we claim that among good vertices, at least half have degree at most $\frac{8g}{\epsilon}$. By contradiction, let $g = |Good(G)|$ and assume that there is a set $Good'(G) \subseteq Good(G)$ containing at least $\frac{g}{2}$ vertices, each of degree greater than $\frac{8g}{\epsilon}$. Then
\[
m \geq \frac{1}{2} \sum_{u \in {Good'(G)}} d(u) > \frac{1}{2} \cdot \frac{g}{2} \cdot \frac{8g}{\epsilon}= \frac{2g^2}{\epsilon}.
\]
The number of edges between good vertices is at most $g \choose 2$ and at least  $\frac{\epsilon m}{2}$. In particular,
\[
{g \choose 2} \geq \frac{\epsilon m}{2} > \frac{ \epsilon}{2} \cdot \frac{2g^2}{\epsilon} = g^2,
\]
a contradiction.

\medskip

It follows from the above that 
\[
\E[Y] \geq \sum_{u \in Good(G)} 2\frac{c(u)}{d(u)^2} \geq \sum_{u \in Good(G)} 2\frac{c(u)}{d(u)}\frac{1}{d(u)}
\]
and, by summing only over the at least $\frac g2$ good vertices of degree at most $\frac{8g}{\epsilon}$, we conclude that
\[
\E[Y] \geq 2 \cdot \frac {g}{2} \cdot  \frac{\epsilon}{24} \cdot \frac{\epsilon}{8g} = \frac{\epsilon^2}{192}.
\]
Recall that $Y$ is the sum of independent boolean variables $X_u$. Let $N$ be the random variable counting the
number of vertices detecting a $K_4$ in $T$ iterations of Algorithm~\ref{al:K4gen}, so $N = T\cdot Y$. For $T = 8 \ln 3 \cdot \frac{192}{\epsilon^2}$, we have that
\[
E[N] \geq  8 \ln 3.
\]
Chernoff bound applies to variable $N$ and, as in the proof of Theorem~\ref{th:p4f}, we conclude that $\Pr[N \leq 4 \ln 3] \leq \frac13$.
Altogether, a copy of $K_4$ is detected with probability at least $\frac{2}{3}$, which completes the proof. 
\end{proof}

\section{Limits  of the BFS Approach}

As it happened with the DFS-based approach, the BFS-based approach fails to generalize to large graphs $H$. Actually, it already fails for $K_5$.  

\begin{theorem}\label{th:neg1}
Let $k\geq 5$. There exists a graph $G$ that is $\epsilon$-far from being $K_k$-free in which any constant number of repetitions of Algorithm~\ref{al:K4gen} fails to detect any copy of $K_k$, with probability at least $1-o(1)$. 
\end{theorem}

\begin{proof}
The family of graphs $BK(k,p)$ constructed in the proof of Theorem~\ref{th:neg}  for defeating Algorithm~\ref{al:PkGen} can also be used to defeat Algorithm~\ref{al:K4gen}. Recall that those graphs have  $n = kp$ vertices,   $m = \binom{k}{2}pp'$ edges, and every vertex is of degree $d= (k-1)p'$, for $p' = p^{1-o(1)}$. Moreover, they have exactly $pp'$ copies of $K_k$, which are pairwise edge-disjoint. $BK(k,p)$ is $\epsilon$-far from being $K_k$-free with $\epsilon = 1/{k \choose 2}$. By Equation~\ref{eq:prK4}, for each copy $K$ of $K_k$, and for every $u \in K$, the probability that $u$ detects $K$ is $d/{d \choose {k-1}} \leq \frac{\alpha}{d^{k-2}}$ for some constant $\alpha > 0$. Therefore,  when running the algorithm $T$ times, it follows from the union bound that the expected number of detected copies of $K_k$ is at most $\frac{\alpha kTpp'}{d^{k-2}}$, which tends to $0$ when $p \to \infty$, for any $k \geq 5$. Consequently, the probability of detects at least one copy of $K_k$ also tends to $0$.
\end{proof}

\section{Conclusion and Further Work}

We have proved that not only the presence of triangles can be tested in $O(1)$ rounds in the \Congest\ model, but also the presence of any  4-node subgraph. On the other hand, we have also proved that the natural approaches used to test such small  patterns fail to detect larger patterns, like the presence of $K_k$ for $k\geq 5$. Interestingly, as mentioned before, the lower bound techniques for the  \Congest\ model are essentially based on reductions to communication complexity problems. Such an approach does not seem to apply easily in the context of distributed testing. The question of whether the presence of large cliques (or cycles) can be tested in $O(1)$ rounds in the \Congest\ model is an intriguing open problem.  

It is worth mentioning that our algorithms generalize to testing the presence of \emph{induced} subgraphs. Indeed, if the input graph $G$ contains at least $\epsilon m$ edge-disjoint \emph{induced} copies of~$H$, for a graph $H$ on four vertices containing a Hamiltonian path, then Algorithm~\ref{al:Pk} detects an induced subgraph $H$ with constant probability (the only difference is that, in the last line of the algorithm, we check for an induced subgraph instead of just a subgraph). Moreover, if the input contains $\epsilon m$ edge-disjoint induced claws (i.e., induced subgraphs $K_{1,3}$), then Algorithm~\ref{al:K4gen} detects one of them with constant probability. Thus, for any connected graph $H$ on four vertices, distinguishing between graphs that do not have $H$ as induced subgraph, and those who have $\epsilon m$ edge-disjoint induced copies of $H$ can be done in $O(1)$ rounds in the \Congest\ model. However, we point out that, unlike in the case of subgraphs, a graph that is $\epsilon$-far from having $H$ as induced subgraph may not have many edge-disjoint induced copies of $H$.

\newpage
\bibliographystyle{plain}

\end{document}